\documentclass{amsart}
\usepackage{fullpage}
\usepackage{graphicx,amsmath,natbib,multicol,txfonts}
\usepackage{amsfonts,amssymb,latexsym,textcomp}
\usepackage{txfonts}
\usepackage{epsfig}
\usepackage{hyperref}
\usepackage{float}
\usepackage{multirow}
\usepackage{rotating}
\usepackage{subfigure}
\usepackage{url}
\usepackage{color}
\usepackage{algorithm}
\usepackage{algpseudocode}

\renewcommand{\Pr}{\mathsf{Pr}}

\newcommand{\E}{\mathsf{E}}
\newcommand{\V}{\mathsf{Var}}

\newcommand{\dd}{\mbox{d}}

\newcommand{\DP}{\mathsf{DP}}

\newcommand{\Gam}{\mathsf{Gam}}

\newcommand{\MED}{\mathsf{MED}}

\newtheorem{lemma}{Lemma}

\begin{document}

\title{Default Bayesian Analysis for the Multivariate Ewens Distribution}
\author{Abel Rodr\'{\i}guez}
\date{}

\begin{abstract}
We derive the Jeffreys prior for the parameter of the Multivariate Ewens Distribution and study some of its properties.  In particular, we show that this prior is proper and has no finite moments.  We also investigate the impact of this default prior on the a priori distribution of the number of species and the a priori probability of discovery of a new species, which are usually employed in subjective prior elicitation.  The effect of the Jeffreys prior for posterior inference is illustrated using examples arising in the context of inference for species sampling models and Dirichlet process mixture models.
\end{abstract}
\keywords{Multivariate Ewens Distribution; Jeffreys Prior; Dirichlet Process; Ewens Sampling Formula}

\maketitle

\section{Introduction}

The Multivariate Ewens Distribution (MED), also known as the Ewens Sampling Formula (ESF) \citep{Ew72,JoKoBa97}, is a probability distribution on the partitions of the set $\{ 1, 2, \ldots, n \}$.  It appears often in genetics as the distribution of the number of distinct alleles in a sample of size $n$ drawn from an infinite idealized population, or as the limiting distribution of other, more general models \citep{Ki78a,Ho87}. More generally, the MED belongs to the class of species sampling models, which describe distributions for exchangeable random partitions  \citep{Al85,Pi95,LiMePr07}.  The MED also appears in the context of Bayesian nonparametric statistics, where it is related to the number of unique values in a random sample of size $n$ taken from a random distribution that follows a Dirichlet process prior \citep{An74}.  Hence, in the context of Dirichlet process mixture models, the MED acts as the prior distribution on the number and size of the clusters imposed by model.

\smallskip

This paper is concerned with Bayesian estimation and prediction under the MED model.  Gamma priors, or mixtures thereof, are commonly used as priors in this context because of the existence of simple Gibbs sampling algorithms based on data augmentation \citep{EsWe95}.  Hyperparameters are elicited by either exploiting their link with the expected number of distinct alleles \citep{EsWe95} or, in the case of nonparametric mixture models, their link with the mean and variance of the observations \citep{WaMa97}.  Alternatively, \cite{CaPa02} and \cite{GrSt04} propose eliciting priors on the probabilities of new alleles, which in turn imply a prior on the parameter of interest.  In either case, elicitation can be difficult because of  the lack of relevant prior information in specific applications.  To deal with the lack of prior information, numerous authors have used very dispersed Gamma priors.  In this paper we derive the Jeffreys prior associated with the MED, show that this prior is proper, and investigate some of its properties.  Because of its invariance to transformations, the Jeffreys prior provides a natural ``default'' or ``non-informative'' alternative to the priors discussed above without a substantial increase in computational complexity.

\smallskip

The remaining of the paper is organized as follows:  Section \ref{se:MED} briefly reviews the Multivariate Ewens Distribution and derives the Jeffreys prior associated with its parameter.   Section \ref{se:properties} discusses some of the properties of the Jeffreys prior.  Section \ref{se:illus} presents some illustrations in the context of species sampling models and Dirichlet process mixture models.  We conclude in Section \ref{se:conclusion} with some brief remarks and future directions.

\section{The Jeffreys Prior for the Multivariate Ewens Distribution}\label{se:MED}

Consider a partition of the set $\{ 1, 2, \ldots, n \}$ into $K \le n$ subsets so that there are $r_j$ subsets of size $j$, where $\sum_{j=1}^{n} r_j = K$, and $\sum_{j=1}^{n} j r_j = n$.  For example, the $n$ elements of the original set might represent individuals being sampled from an infinite population, while the subsets into which they are divided could be interpreted as the species to which these individuals belong.  The Multivariate Ewens Distribution (MED) assigns such a partition a probability 
\begin{align}\label{eq:MED1}
p(K, r_1, \ldots, r_n \mid \beta) = \frac{\Gamma(\beta)\Gamma(n+1)}{\Gamma(\beta + n)} \,\beta^{K} \,   \prod_{j=1}^{n} \frac{1}{j^{r_j} \Gamma(r_j + 1)}  ,
\end{align}
where $0 < \beta < \infty$ is a parameter controlling the shape of the distribution.  

\smallskip

The partitions associated with the MED can alternatively be described in terms of a sequence of exchangeable indicators $\xi_1, \ldots, \xi_n$ such that $\xi_i = k$ if the $i$ individual in the population belong to species $k$.  Assuming that the species are labeled consecutively between at 1 and $K$,  and letting $m_k = \sum_{i=1}^{n} I(\xi_i = k)$ be the number of individuals in species $k$, then
\begin{align}\label{eq:MED2}
p(\xi_1, \ldots, \xi_n \mid \beta) = p(K, m_1, \ldots, m_K \mid \beta) = \frac{\Gamma(\beta)}{\Gamma(\beta + n)} \, \beta^{K} \,\prod_{k=1}^{K} \Gamma(m_k).
\end{align}

\smallskip

From \eqref{eq:MED2} we can compute the probability mass function associated with the species of a new individual
$$
p(\xi_{n+1} = k \mid \xi_{1}, \ldots, \xi_{n}, \beta) = \begin{cases}
\dfrac{ p(K, m_1, \ldots, m_k +1, \ldots, m_K \mid \beta) }{ p(K, m_1, \ldots, m_k, \ldots, m_K \mid \beta) } = \dfrac{m_k}{\beta + n}  & k \le K \\
& \\
\dfrac{ p(K+1, m_1, \ldots, m_k, \ldots, m_K, 1 \mid \beta) }{ p(K, m_1, \ldots, m_k, \ldots, m_K \mid \beta) } = \dfrac{\beta}{\beta + n} & k = K+ 1 . \\
\end{cases}
$$
This sequence of predictive distributions is sometimes called the Chinese restaurant process.

\smallskip

We are interested in estimating the parameter $\beta$ based on either an observed sample $\xi_1,\ldots,\xi_n$, or the sufficient statistic $K$.  
The following lemma provides an expression for a natural default prior for $\beta$.
\begin{lemma}
For $n \ge 2$, the Jeffreys prior associated with \eqref{eq:MED1} and \eqref{eq:MED2} is given by
\begin{align}\label{eq:jefpMED}
\pi_n^{J} (\beta) & \propto \sqrt{\frac{1}{\beta}\sum_{j=1}^{n-1}\frac{j}{(\beta+j)^2}}   .
\end{align}
\end{lemma}

\begin{proof}
By definition, $\pi_n^{J} (\beta)  \propto \left|\mathcal{I}(\beta)\right|^{1/2}$ where $\mathcal{I}(\beta) = -\E\left[\frac{\dd^2}{\dd\beta^2} \log \left\{ p (K, m_1, \ldots, m_K \mid \beta) \right\} \right]$ is the Fisher information associated with $\beta$. Now, 
\begin{align*}
- \E \left[\frac{\dd^2}{\dd \beta^2}  \log \left\{ p (K,m_1, \ldots, m_K \mid \beta) \right\} \right] &= - \psi'(\beta) + \psi'(\beta+n) + \frac{\E \{ K \} }{\beta^2}  ,
\end{align*}
where $\psi'$ denotes the trigamma function \citep{AbSt}.  Now, using the facts that $\E \{ K \} =  \sum_{j=0}^{n-1} \frac{\beta}{\beta + j}$ (e.g., see \citealp{An74}) and $\psi'(\beta+n) = \psi'(\beta) - \sum_{j=0}^{n-1}\frac{1}{(\beta+j)^2}$ (e.g., see \citealp{AbSt}) we get
\begin{align*}
\mathcal{I}_E(\beta) &= - \sum_{j=0}^{n-1}\frac{1}{(\beta+j)^2} + \frac{1}{\beta^2}\sum_{j=0}^{n-1} \frac{\beta}{\beta + j} = \frac{1}{\beta}\sum_{j=0}^{n-1}\frac{j}{(\beta+j)^2} = \frac{1}{\beta}\sum_{j=1}^{n-1}\frac{j}{(\beta+j)^2}  ,
\end{align*}
which directly leads to \eqref{eq:jefpMED}.
\end{proof}

In particular, note that, for $n=2$ (the smallest sample size containing information about $\beta$), the Jeffreys prior on $\beta$ corresponds to a standard Cauchy prior on $\nu = \beta^{1/2}$.

\section{Properties of the Jeffreys Prior for the MED}\label{se:properties}

Surprisingly, the Jeffreys prior is proper.  Indeed, note that for all $n \ge 2$
\begin{align*}
\int_{0}^{\infty} \sqrt{\frac{1}{\beta}\sum_{j=1}^{n-1}\frac{j}{(\beta+j)^2}} \; \dd \beta \le \int_{0}^{\infty} \sqrt{\frac{1}{\beta}\sum_{j=1}^{n-1}\frac{j}{(\beta+1)^2}} \; \dd \beta    & =  \sqrt{\frac{n(n-1)}{2}}  \int_0^\infty \frac{1}{\beta^{3/2}+\beta^{1/2}} \; \dd\beta \\ &=   \sqrt{\frac{n(n-1)}{2}} \int_{0}^{\infty} \frac{1}{u^2+1}\;\dd u = \pi \sqrt{\frac{n(n-1)}{2}}
\end{align*}

\smallskip

Although the normalizing constant $C(n) = \int_{0}^{\infty} \sqrt{\frac{1}{\beta}\sum_{j=1}^{n-1}\frac{j}{(\beta+j)^2}} \; \dd \beta$ is not available in closed form, it is easily evaluated numerically using quadrature methods.  The prior is decreasing, which can be easily verified, 
$$
\frac{\dd}{\dd \beta}\log\{ \pi_n^{J} (\beta) \}  = - \frac{1}{2\beta} - \frac{\sum_{j=1}^{n-1} \frac{j}{(\beta + j)^3}}{\sum_{j=1}^{n-1} \frac{j}{(\beta + j)^2}} < 0  .
$$

\smallskip

On the other hand, $\pi_n^{J}(\beta)$ is {\it log-convex}.  Indeed,
$$
\frac{\dd^2}{\dd \beta^2}\log\{ \pi_n^{J} (\beta) \}  = \frac{1}{2 \beta^2} + \frac{   3\left\{ \sum_{j=1}^{n-1} \frac{j}{(\beta + j)^4} \right\} \left\{ \sum_{j=1}^{n-1} \frac{j}{(\beta + j)^2} \right\} - 2\left\{ \sum_{j=1}^{n-1} \frac{j}{(\beta + j)^3} \right\}^2   }{ \left\{ \sum_{j=1}^{n-1} \frac{j}{(\beta + j)^2} \right\}^2 } > 0 .
$$

\smallskip

To show this, note that the first term is clearly positive over the support of the prior.  For the second term,
\begin{align*}
3\left\{ \sum_{j=1}^{n-1} \frac{j}{(\beta + j)^4} \right\} \left\{ \sum_{j=1}^{n-1} \frac{j}{(\beta + j)^2} \right\} - 2\left\{ \sum_{j=1}^{n-1} \frac{j}{(\beta + j)^3} \right\}^2 
  & = \sum_{j=1}^{n-1} \sum_{i=1}^{n-1} \left\{  \frac{ 3ij }{(\beta+j)^4 (\beta+i)^2} - \frac{ 2ij }{(\beta+j)^3 (\beta+i)^3} \right\} \\
  & = \sum_{j=1}^{n-1} \sum_{i=1}^{n-1} \left\{  \frac{ ij [3(\beta+i) - 2(\beta+j)] }{(\beta+j)^4 (\beta+i)^3}  \right\} \\
  & \ge \frac{1}{(\beta + n)^7} \sum_{j=1}^{n-1} \sum_{i=1}^{n-1} \left\{ ij \beta + 3i^2j - 2ij^2 \right\} \\
  & = \frac{1}{(\beta + n)^7} \left\{ \beta \sum_{j=1}^{n-1} \sum_{i=1}^{n-1} ij + \sum_{j=1}^{n-1} \sum_{i=1}^{n-1} i^2j \right\} > 0.
\end{align*}

\smallskip

Figure \ref{fi:jefpriorgraph} presents graphs of $\pi_n^{J}$ for different values of $n$.  
\begin{figure}[!h]
\centering  \includegraphics[width=.45\textwidth]{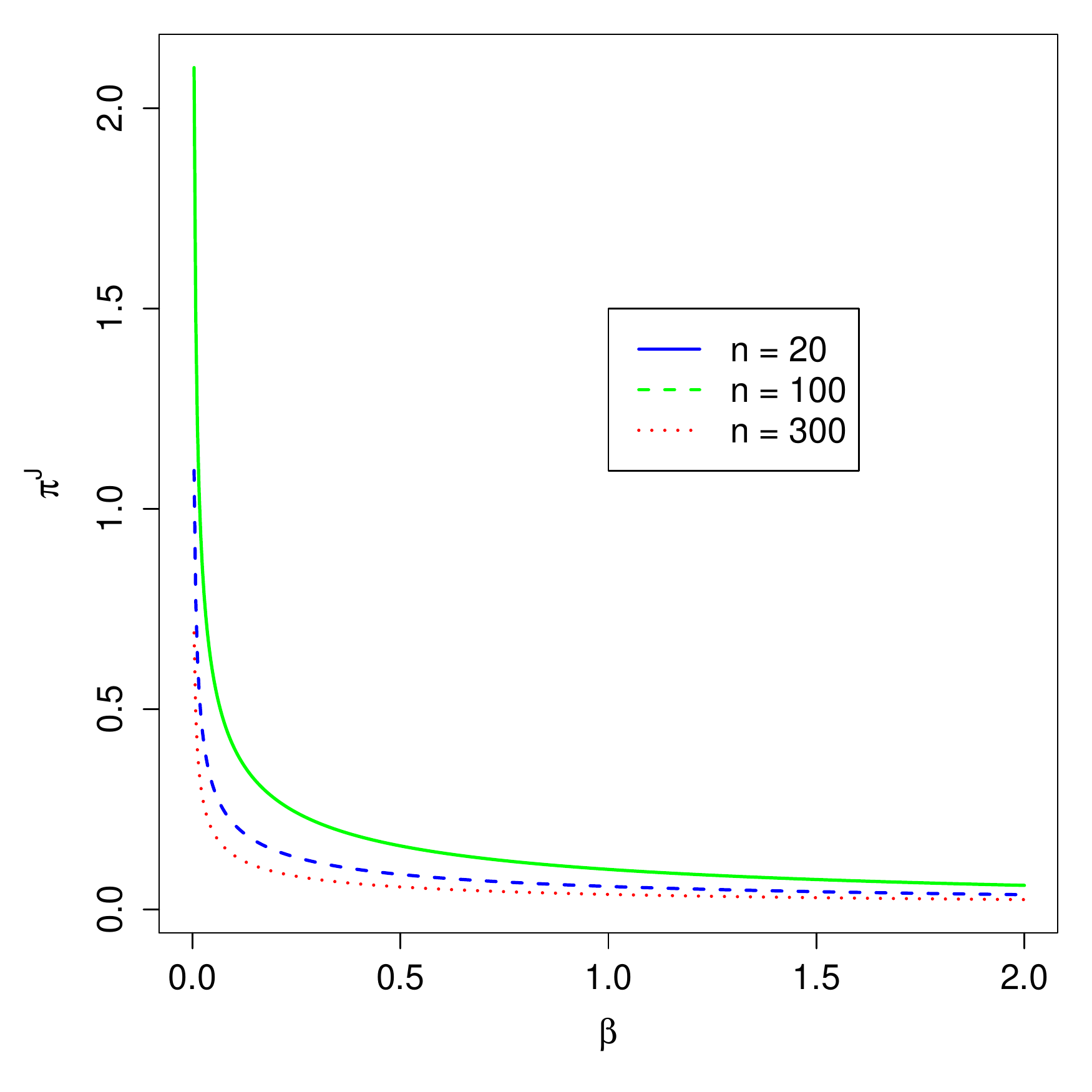}
\caption{Density of the Jeffreys prior associated with the MED for some selected values of $n$.} \label{fi:jefpriorgraph}
\end{figure}

\smallskip

The moments of $\pi_n^J(\beta)$ do not exist.  Indeed, 
\begin{align*}
\int_{0}^{\infty} \beta^{s} \sqrt{\frac{1}{\beta}\sum_{j=1}^{n-1}\frac{j}{(\beta+j)^2}} \; \dd \beta > \int_{0}^{\infty} \beta^{s-1/2} \sqrt{\sum_{j=1}^{n-1}\frac{j}{(\beta+n)^2}} \; \dd \beta    & =   \sqrt{\frac{n(n-1)}{2}} \int_{0}^{\infty} \frac{u^{2s}}{u^2+n}\;\dd u = \infty
\end{align*}
for any $s \ge 1$.  Instead, consider the median of $\pi_n^J(\beta)$ as a function of $n$.  Again, although no closed-form expression is available for this median, it can be computed numerically.  Figure \ref{fi:medianjefprior} suggests that the median grows more or less linearly with the sample size $n$, with $\mathsf{Med}\{ \beta \} \approx 0.36 n + 1$ for $n \le 400$.
\begin{figure}[!h]
\centering  \includegraphics[width=.45\textwidth]{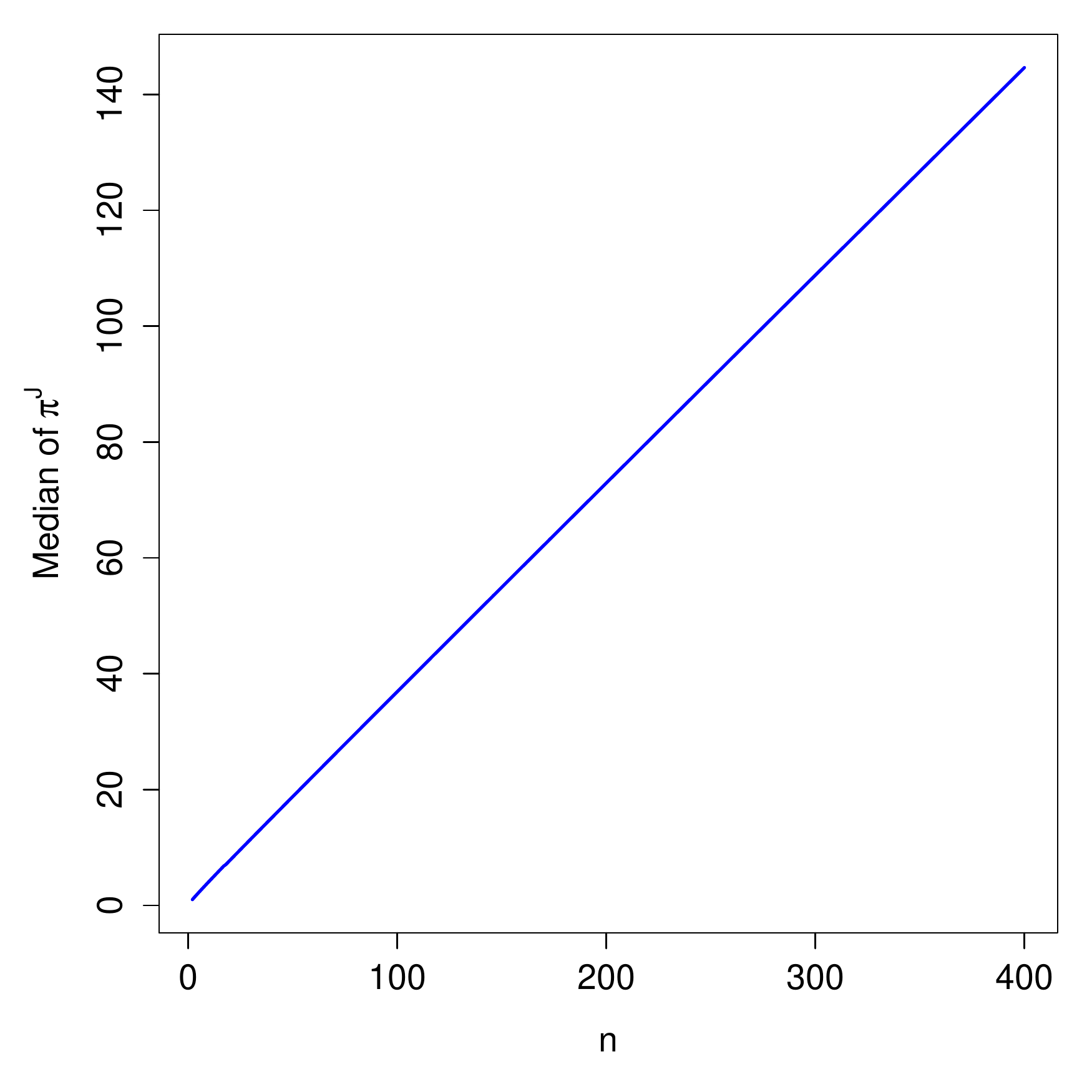}
\caption{Median of $\pi_n^{J}(\beta)$ as a function of the sample size $n$.} \label{fi:medianjefprior}
\end{figure}

\smallskip

Alternatively, we can consider the impact of the Jeffreys prior on the prior distribution of the number of species by computing prior moments such as the prior mean,
\begin{align*}
\E\{ K \mid n \} &= \E\{ \E(K \mid \beta, n) \}\int_{0}^{\infty} \left(\sum_{j=0}^{n-1} \frac{\beta}{\beta + j} \right) \pi_n^{J}(\beta) \dd \beta,   
\end{align*}
and the prior variance
\begin{align*}
 \V\{ K \mid n \} 
 & = \int_{0}^{\infty} \sum_{j=0}^{n-1} \left( \frac{\beta j}{\{ \beta + j \}^2} \right) \pi_n^{J}(\beta) \dd \beta + \int_0^{\infty}\left( \sum_{j=0}^{n-1} \frac{\beta}{\beta+j} \right)^2  \pi_n^{J}(\beta) \dd \beta - \left\{ \int_0^{\infty} \left( \sum_{j=0}^{n-1} \frac{\beta}{\beta + j} \right) \pi_n^{J}(\beta) \dd \beta \right\}^2.
\end{align*}

\smallskip

The left panel in Figure \ref{fi:expec_stdev} presents graphs for the prior mean and prior standard deviation as functions of $n$. Just like the median of $\pi_n^{J}(\beta)$, both of these quantities grow almost linearly with $n$.  Moreover, note that $\E\{ K \} \approx n/2$, and that the prior standard deviation grows somewhat more slowly than $\E\{ K \}$.  More generally, we can compute the marginal prior probability distribution on the number of species,
$$
\Pr(K = k \mid n) = \int \Pr(K = k \mid \beta, n) \pi_n^{J}(\beta) \dd \beta = |S(n,k)|  A(n,n,k)
$$
where
\begin{align}\label{eq:Ank}
A(n,m,k) &= \int_0^{\infty} \beta^{k} \frac{\Gamma(\beta)}{\Gamma(\beta + n)} \pi_m^{J}(\beta) \dd \beta.
\end{align}
and $|S(n,k)|$ is the absolute value of the Stirling number of the second kind.  For illustrative purposes, Figure \ref{fi:dist100} shows the marginal distribution $\Pr(K=k \mid n=100)$. The effect of the Jeffrey's prior is striking; the resulting distribution is U-shaped and roughly symmetric around $n/2$ (which is compatible with our observations about $\E(K \mid n)$).  Hence, under this prior the model prefers a priori either a very small or a very large number of species, with values around the mean/median actually having very low prior probabilities.
\begin{figure}[!h]
\centering  
{\includegraphics[width=.455\textwidth]{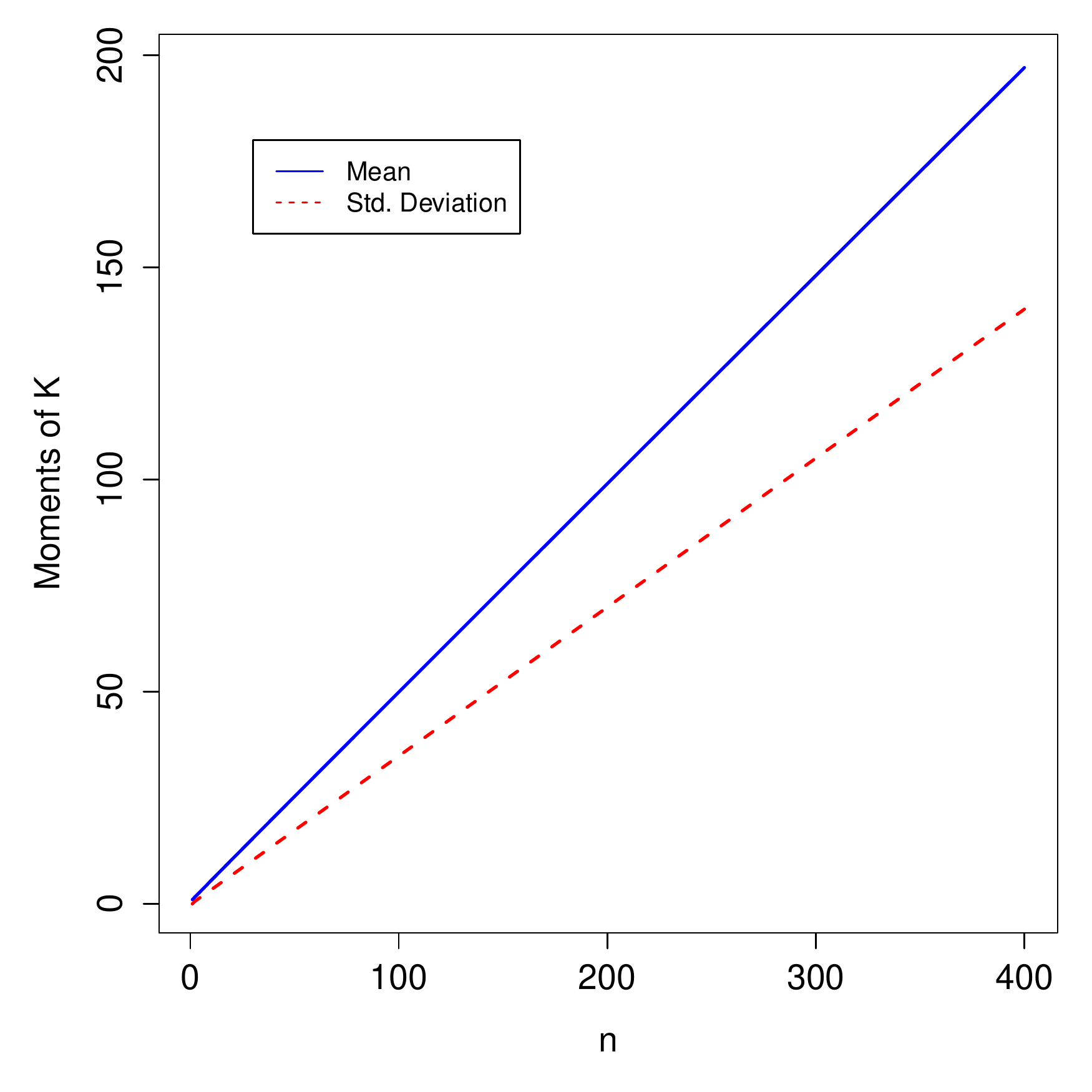} \label{fi:expec_stdev}}
{\includegraphics[width=.455\textwidth]{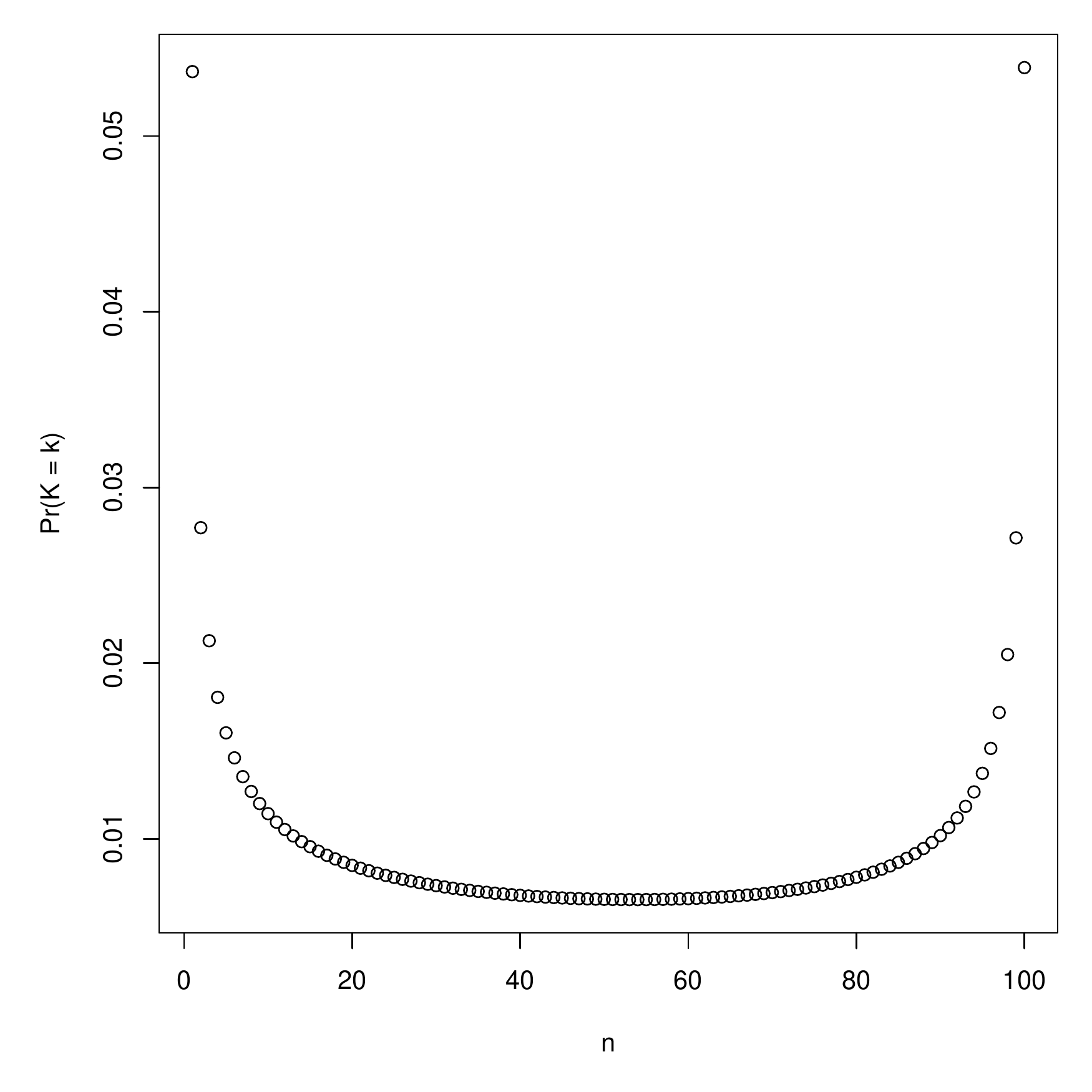} \label{fi:dist100}}
\caption{Left panel: expectation and standard deviation of the number of unique alleles, $E\{ K \mid n\}$ and $\V\{ K \mid n\}$ as a function of the sample size $n$.  Note that they both appear to grow linearly with $n$.  Right panel: full marginal distribution $p(K= k \mid n) = \int_{0}^{\infty} p(K=k \mid \beta)\pi_n^{J}(\beta) \dd \beta$ for $n=100$.} 
\end{figure}

\smallskip

Yet another direction is to consider the effect of $\pi_n^J(\beta)$ on the priori probability of discovery of a new species, $\eta = \beta/(\beta + n + 1)$.  In particular, Figures \ref{fi:expec_stdev} and \ref{fi:dist100} present graphs of the expected value and the variance of $\eta$,
\begin{align*}
\E \left\{ \eta \mid n \right\} &= \int_{0}^{\infty} \left( \frac{\beta}{\beta + n + 1} \right) \pi_n^{J}(\beta) \dd \beta  ,  &
\V \left\{ \eta \mid n \right\} = \int_{0}^{\infty} \left( \frac{\beta}{\beta + n + 1} \right)^2 \pi_n^{J}(\beta) \dd \beta -  \left\{ \int_{0}^{\infty} \left( \frac{\beta}{\beta + n + 1} \right) \pi_n^{J}(\beta) \dd \beta \right\}^2,
\end{align*}
as functions of $n$.  Note that the values of both of these summaries are quite stable, with the expected probability of discovery of a new species varying between 0.38 and 0.40 over the range considered here.

%
%
%
%
\begin{figure}[!h]
\centering  
{\includegraphics[width=.455\textwidth]{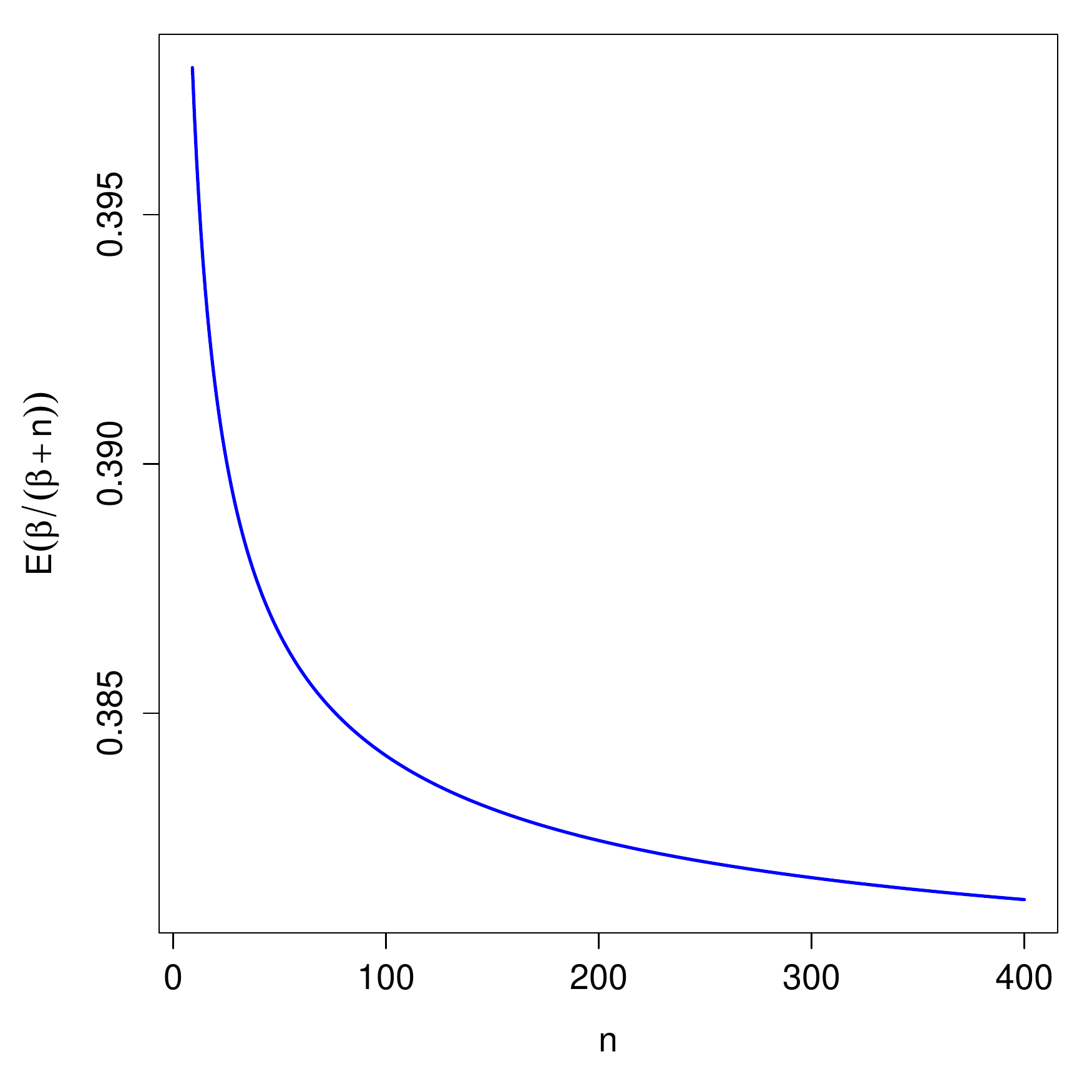} \label{fi:newspec_exp}}
{\includegraphics[width=.455\textwidth]{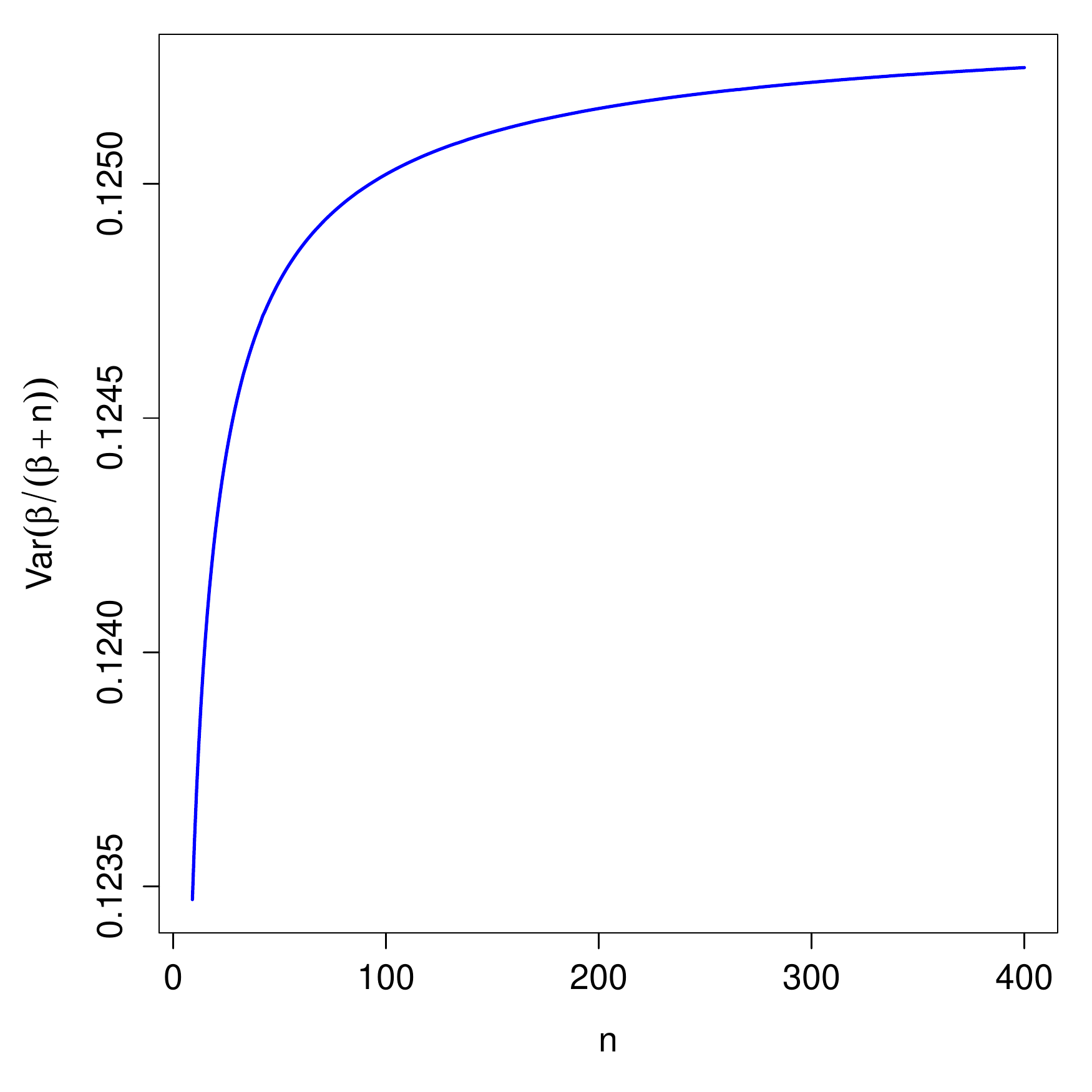} \label{fi:newspec_var}}
\caption{Expectation (left panel) and variance (right panel) of the probability of discovery a new species, $\E\left\{ {\beta}/{(\beta + n + 1)} \right\}$, as a function of the sample size $n$.} 
\end{figure}

\section{Illustrations}\label{se:illus}

\subsection{Species sampling models}\label{se:illustrationSSM}

In this subsection we consider the use of our Jeffreys prior for inference under the MED model.  We consider first a series of simulations where data is generated under a MED with true parameters $\beta_T = 1$ or $\beta_T = 2$.  We consider two different scenarios, corresponding to $n=100$ and $n=1,000$ individuals, and for each of these scenarios we study the frequentist properties of Bayesian interval estimators generated under two different default priors:  the Jeffreys priors derived in this paper, and a diffuse Gamma with shape parameter 0.001 and rate 0.001 (which has mean 1).  The choice of a diffuse Gamma prior centered around $\beta = 1$  reflects standard practice in the literature.

\smallskip

Markov chain Monte Carlo algorithms were used to obtain samples from the posterior distribution under each prior.  All inferences are based on 100,000 iterations of the chain obtained after a burn-in period of 5,000 iterations.  For the Jeffreys prior, we employed a random walk Metropolis-Hasting algorithm with Gaussian proposals for $\log\beta$; the variance of the proposal was $\tau^2=0.05$ for $n=1,000$ and $\tau^2=1$ for $n=100$, which resulted in average acceptance rates of 70\% and 60\% respectively.  For the Gamma prior, we employed the latent variable approach discussed in \cite{EsWe95}, which does not require any tunning parameter.

\begin{table}
\begin{tabular}{|c|c|c|l|cc|} \hline
\multirow{2}{*}{$\beta_T$}  &  \multirow{2}{*}{$n$}     &   \multirow{2}{*}{$\alpha$}   &   \multirow{2}{*}{Summary}   &  \multicolumn{2}{c|}{Priors}   \\
&       &      &      &  Jeffreys  & $\Gam\left(0.001, 0.001\right)$   \\ \hline\hline
\multirow{8}{*}{1}  &  \multirow{4}{*}{100}     &  \multirow{2}{*}{0.90} & Coverage                              &   0.89   & 0.86      \\ 
                                     &                                          &                                       & \multirow{1}{*}{Width}        &   1.85 (0.59)   &   1.74 (0.62)     \\ \cline{3-6}
                                     &                                          &  \multirow{2}{*}{0.95} & Coverage                             &  0.95    &   0.89    \\  
                                     &                                          &                                       & \multirow{1}{*}{Width}         &   2.24 (0.71)    &  2.11 (0.74)   \\  \cline{2-6}
                                     &  \multirow{4}{*}{1000}   &  \multirow{2}{*}{0.90} & Coverage                              &   0.90   &   0.86    \\ 
                                     &                                          &                                       & \multirow{1}{*}{Width}        &   1.40 (0.34)   &   1.34 (0.35)   \\ \cline{3-6}
                                     &                                          &  \multirow{2}{*}{0.95} & Coverage                             &   0.94     &  0.94     \\  
                                     &                                          &                                       & \multirow{1}{*}{Width}         &   1.67 (0.38)   &   1.61 (0.39)     \\  \hline \hline
\multirow{8}{*}{2}  &  \multirow{4}{*}{100}     &  \multirow{2}{*}{0.90} & Coverage                              &   0.91   & 0.86      \\ 
                                     &                                          &                                       & \multirow{1}{*}{Width}        &   2.84 (0.75)   &   2.74 (0.76)     \\ \cline{3-6}
                                     &                                          &  \multirow{2}{*}{0.95} & Coverage                             &  0.95    &   0.94    \\  
                                     &                                          &                                       & \multirow{1}{*}{Width}         &   3.43 (0.90)    &  3.32 (0.91)   \\  \cline{2-6}
                                     &  \multirow{4}{*}{1000}   &  \multirow{2}{*}{0.90} & Coverage                              & 0.91     &  0.88     \\ 
                                     &                                          &                                       & \multirow{1}{*}{Width}        &   2.07 (0.37)   &  2.02 (0.38)    \\ \cline{3-6}
                                     &                                          &  \multirow{2}{*}{0.95} & Coverage                             &    0.94    &   0.95    \\  
                                     &                                          &                                       & \multirow{1}{*}{Width}         &   2.48 (0.44)     &    2.42 (0.45)    \\  \hline
\end{tabular}
\vspace{0.2cm}
\caption{Results of a simulation study to explore the frequentist coverage and expected length of $100\alpha\%$ credible intervals for $\beta$ generated under three different priors.}\label{ta:simstudy1}
\end{table}


\smallskip

Table \ref{ta:simstudy1} presents empirical coverage probabilities and interval lengths for symmetric credible sets constructed.  These summaries were estimated on the basis of $2,000$ randomly generated datasets.  As expected, the frequentist coverage probability of 90\% and 95\% symmetric credible intervals under the Jeffreys prior seems to coincide with their nominal posterior probability.  On the other hand, the dispersed Gamma prior produces intervals that are somewhat tighter than the Jeffreys prior, but which they tend to have lower empirical coverage rates (particularly, for 90\% coverage).

\smallskip

As a second illustration, we consider a real dataset discussed in \cite{MaLi02}, \cite{Ma04}, and \cite{LiMePr07}.  The data consists of a $n=2586$ randomly selected expressed sequence tags taken form a large cDNA library made from the 0 mm  to 3 mm buds of tomato flowers.  The number of distinct tags observed in this sample is $K=1825$.  Table \ref{ta:tags} presents estimates of the MED parameter $\beta$ under the same two priors discussed in the simulation study.  The same MCMC algorithms described in the simulation study were used in this analysis.  Note that, although inferences for $\beta$ differ somewhat among the two priors, inferences for the probability of discovery of a new species, $\eta = \beta/(\beta+n + 1)$, are almost identical in both cases.
\begin{table}
\begin{tabular}{|l|cc|cc|} \hline
& \multicolumn{2}{c|}{$\beta$}  &  \multicolumn{2}{c|}{$\eta = \beta/(\beta + n + 1)$} \\ \cline{2-5} 
& Post. mean  &  95\% Cred. Interval   &   Post. mean  &  95\% Cred. Interval \\ \hline\hline
Jeffreys & 2763.3  &   (2526.9,  3015.8)  & 0.516  &   (0.494, 0.538) \\
$\Gam\left(0.001, 0.001\right)$  & 2751.3  &   (2518.1,  3002.4)   & 0.515  &   (0.493, 0.537) \\ \hline
\end{tabular}
\vspace{0.2cm}
\caption{Posterior credible inferences for $\beta$ and $\eta = \beta/(\beta + n + 1)$ (the probability of discovery of a new species) under two different ``non-informative'' prior distributions for the sequencing tag data.}\label{ta:tags}
\end{table}

\subsection{Dirichlet process mixture models}\label{se:illustrationDPM}

The Dirichlet process (DP) \citep{Fe73,An74,Se94} defines a prior distribution on the space of discrete measures and has been widely used in the context of nonparametric Bayesian inference.  However, because of the discrete nature of distributions, the DP is not typically used to model the data directly, but as a prior for the mixing distribution in a kernel convolution.  In that case, the data generating process for an independent and identically distributed sample $y_1, \ldots, y_n$ is assumed to be
\begin{align*}
y_i \mid G &\sim \int p(y_i \mid \theta) G(\dd \theta),   & G & \mid \beta \sim \DP(\beta, G_0),
\end{align*}
where $\DP(\beta, G_0)$ denotes a Dirichlet process prior with centering measure $G_0$ and precision parameter $\beta$, and $p(y_i | \theta)$ is a kernel indexed by the finite-dimensional parameter vector $\theta$.  Such a model can alternatively be described in terms of a series of partition indicators $\xi_1, \ldots, \xi_n$ and component-specific parameters $\vartheta_1, \vartheta_2, \ldots, $ such that
\begin{align*}
y_i \mid \xi_i, \vartheta_1, \vartheta_2, \ldots &\sim p(y_i \mid \vartheta_{\xi_i}) ,   & \xi_1, \ldots, \xi_n \mid \beta &\sim \MED(\beta)   ,    & \vartheta_k &\sim  G_0,
\end{align*}
where $\MED(\beta)$ represents the multivariate Ewens distribution with parameter $\beta$.

\smallskip

We ran two simulation studies to investigate the impact of the (marginal) Jeffreys prior on posterior inferences for the DP mixture model.  First, a dataset consisting of 50 observations was generated from a negative binomial distribution with mean 20 and variance 220, and a DP mixture of Poisson kernels with an unknown precision parameter $\beta$ and a Gamma baseline measure was fitted to this data (the baseline measure was selected so that it had mean 20 and variance 200).  Note that, because the negative binomial can be represented as a scale mixture of Poissons, the true data generating process corresponds to the limit of the Poisson DP mixture prior when $K= n=50$ (which can be obtained by letting $\beta \to \infty$).

\smallskip

We considered two different prior distribution for $\beta$, namely, the Jeffreys prior introduced in this paper and the ``non-informative'' $\Gam(0.001, 0.001)$.  A variant of the collapsed Gibbs samplers described in \cite{Ne00} was used to fit the model.  In the case of the Jeffreys prior, we used a version of the algorithm that integrates over $\beta$, so that the posterior full conditional distribution for $\xi_i$ is given by 
\begin{align}\label{eq:PU}
p(\xi_{i} = k \mid \xi_{1}, \ldots, \xi_{i-1}, \xi_{i+1}, \ldots, \xi_{n}) = \begin{cases}
 m^{-}_k \dfrac{ A(n, n, K^{-}) }{ A(n-1, n, K^{-}) } p(y_i | y^{-}, G_0) & k \le K^{-} \\
& \\
\dfrac{ A(n, n, K^{-}+1) }{ A(n-1, n, K^{-}) } p(y_i | G_0)  & k = K^{-}+ 1  , \\
\end{cases}
\end{align}
where $A(n, m, k)$ was defined in \eqref{eq:Ank} and the negative exponent denotes the appropriate quantities computed after excluding observation $i$.

\smallskip

We focused our analysis on the posterior distribution of the number of occupied mixture components $K$.  The posterior mean for $K$ was very similar in both cases, 9.1309 under the Jeffreys prior and 9.0468 under the diffuse Gamma prior, with configurations including more than 16 mixture components having negligible posterior probability.  Having a number of occupied clusters that is smaller than $n$ is not really surprising; because the Dirichlet process prior strongly favors clustering, we expect the model to underestimate the number of mixture components.  What is really interesting is that the Jeffreys prior seems to favor a larger number of clusters than the Gamma prior.  Indeed, the posterior distribution of $K$ under the Jeffreys prior seems to be stochastically greater than the posterior distribution under the Gamma prior (the same phenomena appeared when we repeated the simulation study with other datasets).  This suggest that the Jeffreys prior does a slightly better job at identifying the true number of components in this case.

\smallskip

Finally, in order to evaluate whether the previous behavior is due to a systemic bias in the Jeffreys prior towards larger values of $K$, we ran a similar experiment where data was generated instead from a Poisson distribution with mean 20.  Hence, in this case $K=1$ corresponds to the truth.  In this case, the posterior distribution for $K$ under both models was identical (up to Monte Carlo error), so systematic bias does not seem to be present.

\section{Concluding remarks}\label{se:conclusion}

To the best of our knowledge, this is the first derivation of the Jeffrey's prior associated with the MED.  Our numerical evaluations suggest that it might represent a reasonable default prior in situations where little prior information is available, including hierarchical models such as nonparametric mixture models based on the Dirichlet process.  However, the Jeffreys prior explicitly depends on the sample size observed.  Hence, any statistical procedures derived under this prior will depend on the stopping rule associated with the experiment; for example, the results will vary depending on whether data is analyzed sequentially or in batches.

\bibliographystyle{bka}
\bibliography{MED}

\end{document}